\documentclass[journal,draftclsnofoot,onecolumn,12pt,twoside]{IEEEtran}
\usepackage{amsmath,amsfonts}
\usepackage{algorithmic}
\usepackage{algorithm}
\usepackage{amssymb}
\usepackage{subfigure}

\allowdisplaybreaks[4]
\usepackage{float} 

\usepackage{color}
\usepackage{array}
\usepackage[caption=false,font=normalsize,labelfont=sf,textfont=sf]{subfig}
\usepackage{textcomp}
\usepackage{stfloats}
\usepackage{url}
\usepackage{bm}
\usepackage{verbatim}
\usepackage{multirow} 
\usepackage{ntheorem}
\usepackage{graphicx}
\usepackage{cite}
\theorembodyfont{\upshape}
\theoremheaderfont{\it}
\qedsymbol{\square}
\theoremseparator{:}
\newtheorem{theorem}{\indent Theorem}

\newtheorem*{proof}{\indent Proof}

\newtheorem{proposition}{\indent Proposition}

\makeatletter

\newcommand{\Rmnum}[1]{\expandafter\@slowromancap\romannumeral #1@}
\makeatother
\begin{document}

\title{Empowering Over-the-Air Personalized Federated Learning via RIS}

\author{Wei~Shi, Jiacheng~Yao, Jindan~Xu,~\IEEEmembership{Member,~IEEE}, Wei~Xu,~\IEEEmembership{Senior~Member,~IEEE}, Lexi Xu, and Chunming~Zhao,~\IEEEmembership{Member,~IEEE}
\thanks{W. Shi, J. Yao, W. Xu, and C. Zhao are with the National Mobile Communications Research Laboratory, Southeast University, Nanjing 210096, China, and are also with the Purple Mountain Laboratories, Nanjing 211111, China (e-mail: \{wshi, jcyao, wxu, cmzhao\}@seu.edu.cn).}
\thanks{J. Xu is with the School of Electrical and Electronics Engineering, Nanyang Technological University, Singapore 639798, Singapore (e-mail: jindan.xu@ntu.edu.sg).}
\thanks{L. Xu is with the Research Institute, China United Network Communications Corporation, Beijing 100048, China (e-mail: davidlexi@hotmail.com).}
}

\maketitle

\begin{abstract}

Over-the-air computation (AirComp) integrates analog communication with task-oriented computation, serving as a key enabling technique for communication-efficient federated learning (FL) over wireless networks. However, AirComp-enabled FL (AirFL) with a single global consensus model fails to address the data heterogeneity in real-life FL scenarios with non-independent and identically distributed local datasets. In this paper, we introduce reconfigurable intelligent surface (RIS) technology to enable efficient personalized AirFL, mitigating the data heterogeneity issue. First, we achieve statistical interference elimination across different clusters in the personalized AirFL framework via RIS phase shift configuration. Then, we propose two personalized aggregation schemes involving power control and denoising factor design from the perspectives of first- and second-order moments, respectively, to enhance the FL convergence. Numerical results validate the superior performance of our proposed schemes over existing baselines.


\end{abstract}

\section{Introduction}
To build ubiquitous intelligence at the edge of wireless networks, federated learning (FL) stands out as a promising distributed learning approach due to its privacy-enhancing characteristic \cite{0,1,gzhu}. In a wireless FL system, multiple distributed devices communicate with a parameter server (PS) via wireless links for collaborative model training~\cite{yjc1,2}. To enhance communication efficiency of wireless FL, over-the-air computation (AirComp) has emerged as a key technique by exploiting the waveform superposition property of multiple access channels. Specifically, AirComp enables fast aggregation of gradients from distributed devices through non-orthogonal multiple access, aligning with FL's requirement of averaging local gradients without necessitating access to individual values~\cite{3}.

Although AirComp-enabled FL (AirFL) offers significant performance gains, it does not address the data heterogeneity in most real-life FL scenarios with non-independent and identically distributed local datasets. Such data heterogeneity hinders the generalization of a single global consensus model. To this end, preliminary works have been made to develop a personalized AirFL framework via clustering algorithms, where different models are trained for different clusters under the orchestration of the PS~\cite{4,5}. However, this personalized framework requires large-scale receiving antennas to combat interference, leading to a significant escalation in hardware cost.

As a cost-effective physical-layer technology, reconfigurable intelligent surface (RIS) has been extensively studied to support various communication applications due to its capability for smart channel reconstruction~\cite{6,sw3}. In this paper, we introduce low-cost RIS to achieve statistical interference elimination across different clusters and facilitate simultaneous multi-cluster computation over-the-air, thereby enhancing the efficiency of personalized AirFL.

\section{System Model}
We consider a personalized AirFL system consisting of $K$ distributed devices, which are partitioned into $M$ ($M\!<\!K$) disjoint clusters $\mathcal{K}_1,\dots,\mathcal{K}_M$. A specific clustering method can be found in~\cite{4}, which is not the focus of this paper. Our goal is to find the optimal personalized model parameters $\mathbf{w}_m\in \mathbb{R}^D$ for each cluster $m\!\in\![M]$ to minimize the loss function
${\mathcal{L}}_m\!\left(\mathbf{w}_m\right)\!=\!\frac{1}{\left|\mathcal{K}_m\right|}\sum_{k\in\mathcal{K}_m}{F_k\!\left(\!\mathbf{w}_m,\mathcal{D}_{k}\!\right)}$,
where $F_k\left(\cdot,\mathcal{D}_{k}\right)$ is the loss function of device $k$ with local dataset $\mathcal{D}_{k}$. 

Distributed stochastic gradient descent (SGD) is adopted to optimize $\mathbf{w}_m$ in an iterative manner. First, at each training round $t$, the PS broadcasts the latest personalized models $\{\mathbf{w}_{m,t}\}_{m\in[M]}$ to each device. Then, based on the clustering mechanism, each device $k\in\mathcal{K}_m$ computes its local gradient $\mathbf{g}_{m,t,k}\in\mathbb{R}^{D}$ based on $\mathbf{w}_m$ and its local dataset $\mathcal{D}_{k}$, and reports it to the PS. Finally, after receiving all the local gradients, the PS calculates the global gradient of cluster $m$ as
\begin{equation}
\mathbf{g}_{m,t} = \frac{1}{\left|\mathcal{K}_m\right|}\sum_{k\in\mathcal{K}_m}\mathbf{g}_{m,t,k},
\label{eq2}
\end{equation}
and updates the personalized model for cluster $m$ through $\mathbf{w}_{m,t+1}\!=\!\mathbf{w}_{m,t}\!-\!\eta_{m,t} \mathbf{g}_{m,t}$, where $\eta_{m,t}$ is a chosen learning rate at the $t$-th training round.
The above steps iterate until a convergence condition is met.

Note that the operation in (\ref{eq2}) requires the PS to sum the local gradients of devices in each cluster separately. By applying AirComp, all devices simultaneously upload the analog signals of local gradients to the PS, achieving summation over-the-air. However, the analog nature of AirFL makes the PS cannot distinguish between the gradients of different clusters.
In the following, we introduce an RIS-enabled personalized AirFL framework to address this challenge. Each cluster is assisted by an RIS with $N$ reflecting elements to help realize the personalized model aggregation. To support simultaneous multi-cluster gradient estimation, at least $M$ receiving antennas are required. Without loss of generality, we consider a PS equipped with $M$ receiving antennas. Then, the received signal at the PS in the $t$-th round, $\mathbf{Y}_t\!=\!\left[\mathbf{y}_{1,t},\mathbf{y}_{2,t},\cdots\!,\mathbf{y}_{M,t}\right]^H\!\in\!\mathbb{C}^{M\!\times\! D}$, is given by
\begin{equation}
\mathbf{Y}_t=\sum_{k=1}^K \sqrt{p_{k}} \left(\sum_{i=1}^{M} \beta_{i,k}\mathbf{H}_{p,i}^H\mathbf{\Theta}_i\mathbf{h}_{i,k}\right) \bar{\mathbf{g}}_{m,t,k}^H +\mathbf{Z}_t,
\label{eq3}
\end{equation}
where $\bar{\mathbf{g}}_{m,t,k}\!\triangleq\!\frac{1}{\sigma_{m,t,k}}(\mathbf{g}_{m,t,k}\!-\!u_{m,t,k}\mathbf{1})$ represents the normalized gradient, $u_{m,t,k}$ and $\sigma_{m,t,k}$ denote the mean and standard deviation of all entries in $\mathbf{g}_{m,t,k}$, $p_{k}$ is the transmit power of device $k$, $\beta_{i,k}$ is the cascaded large-scale fading coefficient from device $k$ to the PS through the $i$-th RIS, $\mathbf{H}_{p,i}\!=\!\left[\mathbf{h}_{p,i,1},\mathbf{h}_{p,i,2},\cdots,\mathbf{h}_{p,i,M}\right]\sim \mathcal{CN}(\mathbf{0},\mathbf{I}_{N}\otimes\mathbf{I}_M)$ and $\mathbf{h}_{i,k}\!\sim\!\mathcal{CN}(\mathbf{0},\mathbf{I}_N)$ denote the small-scale fading channel from the $i$-th RIS to the PS and device $k$ to the $i$-th RIS, respectively, $\mathbf{Z}_t\!=\!\left[\mathbf{z}_{1,t},\mathbf{z}_{2,t},\cdots,\mathbf{z}_{M,t}\right]^H$ is additive white Gaussian noise whose entries follow $\mathcal{CN}(0,\sigma^2)$, $\mathbf{\Theta}_i\!\triangleq\!{\rm diag}\left\{{\rm e}^{j\theta_{i\!,\!1}},\ldots,{\rm e}^{j\theta_{i\!,\!n}},\ldots,{\rm e}^{j\theta_{i\!,\!N}}\right\}$ is the reflection matrix of the $i$-th RIS, and $\theta_{i,n}\!\in\![0,2\pi)$ is the phase shift introduced by the $n$-th RIS reflecting element. 
% We adopt the typical Rayleigh fading channel model to characterize the small scale fadings, i.e., $\mathbf{H}_{p,i}\sim \mathcal{CN}(\mathbf{0},\mathbf{I}_{N}\otimes\mathbf{I}_M)$ and $\mathbf{h}_{i,k}\sim \mathcal{CN}(\mathbf{0},\mathbf{I}_N)$.
Then, based on the signal $\mathbf{y}_{m,t}$ at the $m$-th receiving antenna, the PS computes an estimated global gradient of cluster $m$ as ${\hat{\mathbf{g}}}_{m,t}\!=\!\frac{\Re\left\{\mathbf{y}_{m,t}\right\}}{\lambda_m}\!+\!\sum_{k\in\mathcal{K}_m }\!\!\frac{u_{m,t,k}}{|\mathcal{K}_m|}\mathbf{1}$, where $\lambda_m\!>\!0$ is a denoising factor introduced by the PS. It is rewritten as
\begin{align}
\hat{\mathbf{g}}_{m,t}=\sum_{k\in\mathcal{K}_m}{{\ell_{m,k}}\bar{\mathbf{g}}_{m,t,k}}+ \sum_{k\in\mathcal{K}_m }\!\frac{u_{m,t,k}}{|\mathcal{K}_m|}\mathbf{1}+\sum_{\substack{1\le m^\prime\le M\\m^\prime\neq m}}\sum_{k^\prime\in\mathcal{K}_{m^\prime}}{\ell_{m,k^\prime}\bar{\mathbf{g}}_{m^\prime,t,k^\prime}}\!+\bar{\mathbf{z}}_{m,t},
\label{eq4}
\end{align}
where $\ell_{m,k}\!=\!\frac{\sqrt{p_k}}{\lambda_m}\sum_{i=1}^{M}{\beta_{i,k}\Re\!\left\{\!\mathbf{h}_{p,i,m}^H\mathbf{\Theta}_i\mathbf{h}_{i,k}\!\right\}\!}$, $\forall m,k$, and $\bar{\mathbf{z}}_{m,t}\!\triangleq\!\frac{\Re\!\{\mathbf{z}_{m,t}\}}{\lambda_m}$ is the equivalent noise.
Note that the estimated gradient is interfered by signals from other clusters and these interference cannot be eliminated since $M<K$. To this end, we propose an RIS phase shift configuration scheme that fortunately eliminates the interference from a statistical perspective in the following theorem.

\begin{theorem} \label{theorem1}
    Statistical interference elimination, i.e., $\mathbb{E}[\ell_{m,k}]\!>\!0$, $\forall k\!\in\!\mathcal{K}_m$ and $\mathbb{E}[\ell_{m,k^\prime}]\!=\!0$, $\forall k^\prime\!\notin\!\mathcal{K}_m$, can be achieved by setting
    \begin{align}\label{RISphase}
    \theta_{m,n}&=-\angle h_{p,m,m,n}^\ast+\angle\sum_{k\in\mathcal{K}_m} h_{m,k,n}^\ast,
    \end{align}
    for $m\!\!\in\!\!\left[M\right]$ and $n\!\!\in\!\!\left[N\right]$, where $h_{p,m,m,n}$ and $h_{m,k,n}$ are the~$n$-th elements of channel vectors $\mathbf{h}_{p,m,m}$ and $\mathbf{h}_{m,k}$, respectively.
\end{theorem}
\begin{proof}
    See Appendix A. $\hfill\blacksquare$
\end{proof}
According to Theorem \ref{theorem1}, we conclude that favorable propagation can be achieved through phase matching using low-cost RIS reflecting elements, thereby eliminating the need for expensive large-scale receiving antennas. After the statistical interference elimination, we focus on joint design of power control and denoising factors from the following two perspectives to enhance the FL convergence.

\emph{1) Unbiased design:} From the perspective of first-order moment, ensuring unbiased gradient estimation is of pivotal significance for guaranteeing FL convergence~\cite{2}. Hence, we consider the following unbiasedness-oriented method.
\begin{proposition}\label{proposition1}
    By setting $p_{k}=\sigma_{m,t,k}^2\beta_{m,k}^{-2}{\zeta_m^2}$, $\forall k \in \mathcal{K}_m$, and $\lambda_m=\frac{\pi N\!\sqrt{\left|\mathcal{K}_m\right|}\zeta_m}{4}$, the gradient estimation in (\ref{eq4}) is unbiased, where $\zeta_m\!\!=\!\!\min\limits_{k\in\mathcal{K}_m}\!\!\!\frac{\sqrt{P_{k}}\beta_{m\!,k}}{\sigma_{m,t,k}\sqrt{D}}$ and $P_k$ is the maximum transmit power.
\end{proposition}
\begin{proof}
    See Appendix B. $\hfill\blacksquare$
\end{proof}

\emph{2) Minimum mean squared error (MMSE) design:} Apart from unbiasedness of the first-order moment, the second-order moment, known as MSE, also plays a decisive role in FL convergence \cite{7}. For any given power control of $p_k$, we derive the optimal denoising factors in closed form in the following proposition.
\begin{proposition}\label{proposition2}
    The optimal denoising factor of cluster~$m$ for minimizing MSE is equal to
    \begin{align}\label{eq6}
    \!\!\lambda_m^*\!=\!|\mathcal{K}_m|\frac{\sum_{i=1}^M \sum_{k\in\mathcal{K}_i}p_k \bar{h}_{m,k}^2\sigma_{i,t,k}^2\!+\!\frac{\sigma^2}{2}}{\sum_{k\in\mathcal{K}_m}{\sqrt{p_k}\bar{h}_{m,k} \sigma_{m,t,k}^3}},
    \end{align}
    where $\bar{h}_{m,k} \triangleq \sum_{i=1}^{M}{\beta_{i,k}\Re\!\left\{\!\mathbf{h}_{p,i,m}^H\mathbf{\Theta}_i\mathbf{h}_{i,k}\!\right\}\!}$.
\end{proposition}
\begin{proof}
See Appendix C. $\hfill\blacksquare$
\end{proof}
Substituting the optimal $\lambda_m^*$, we formulate a power control optimization problem for minimizing the sum MSE, which can be solved via typical optimization methods and please refer to Appendix C for details.

To summarize, we conclude the proposed RIS-enabled personalized AirFL approach in Algorithm \ref{alg0}.
\begin{algorithm}\footnotesize 
\caption{roposed RIS-enabled personalized AirFL approach} \label{alg0}
\begin{algorithmic}[1]  
\REPEAT
\STATE The PS broadcasts the latest personalized models $\{\mathbf{w}_{m,t}\}_{m\in[M]}$ to each device.
\FOR{ each device $k=1,2,\cdots,K$}
\STATE Identify its cluster $\mathcal{K}_m$. 
\STATE Computes its local gradient $\mathbf{g}_{m,t,k}$ based on $\mathbf{w}_m$ and its local dataset $\mathcal{D}_{k}$.
\STATE Normalize its local gradient as $\bar{\mathbf{g}}_{m,t,k}\!\triangleq\!\frac{1}{\sigma_{m,t,k}}(\mathbf{g}_{m,t,k}\!-\!u_{m,t,k}\mathbf{1})$.
\STATE Upload the mean $u_{m,t,k}$, standard deviation $\sigma_{m,t,k}$, and its cluster identity $m$ to the PS.
\ENDFOR
\STATE The PS configures $N$ RISs according to the specific clusters $\left\{\mathcal{K}_m\right \}_{m=1}^M$ and $\theta_{m,n}=-\angle h_{p,m,m,n}^\ast+\angle\sum_{k\in\mathcal{K}_m} h_{m,k,n}^\ast$.
\STATE  The PS calculates the power control $p_k$ for each device and the denoising factor $\lambda_m$ for each cluster according to the unbiased strategy or the MMSE design, and then feedback the selected power $\left\{p_k\right\}_{k=1}^K$ to each device.
\STATE Each device simultaneously upload its local gradient to the PS based on the predetermined transmit power $p_k$.
\STATE Based on the received signal, the PS computes an estimated global gradient of cluster $m$ as ${\hat{\mathbf{g}}}_{m,t}\!=\!\frac{\Re\left\{\mathbf{y}_{m,t}\right\}}{\lambda_m}\!+\!\sum_{k\in\mathcal{K}_m }\!\!\frac{u_{m,t,k}}{|\mathcal{K}_m|}\mathbf{1}$.
\STATE The PS updates the personalized model for cluster $m$ through $\mathbf{w}_{m,t+1}\!=\!\mathbf{w}_{m,t}\!-\!\eta_{m,t} \hat{\mathbf{g}}_{m,t}$.
\STATE Set $t=t+1$.
\UNTIL{Convergence}
\end{algorithmic} 
\end{algorithm}

\section{Numerical Results}
Assuming that the distances between the PS and each RIS are $200~{\rm m}$, and all the devices in each cluster~$m\!\in\![M]$ are uniformly distributed within a disk of radius $300~{\rm m}$ centered at the $m$-th RIS. The path loss exponent for all the links is $2.2$. 

Fig.~\ref{fig1} depicts the normalized MSE (NMSE) as a function of $N$ for different values of $P_k$. It is observed that as $P_k$ increases, the improvement in NMSE performance is marginal. This is due to the fact that an increase in transmit power amplifies not only the useful signals but interference. Furthermore, the NMSE of our proposed designs decreases linearly with large $N$ on a log-log scale. This phenomenon becomes more obvious as $P_k$ increases, owing to the diminishing impact of noise error. In addition, the NMSE curves with corrupted RIS phase shifts (i.e., 1-bit phase noise) are presented to validate the robustness of our proposed designs. The baseline utilizing random RIS phase shifts fails to obtain any effective performance enhancements as $N$ and $P_k$ increases, which demonstrates the importance of RIS phase shift configuration in Theorem \ref{theorem1}.

\begin{figure}[!t]
    \setlength{\abovecaptionskip}{0pt}
    \setlength{\belowcaptionskip}{0pt}
    \centering
    \includegraphics[width=120 mm,height=8cm]{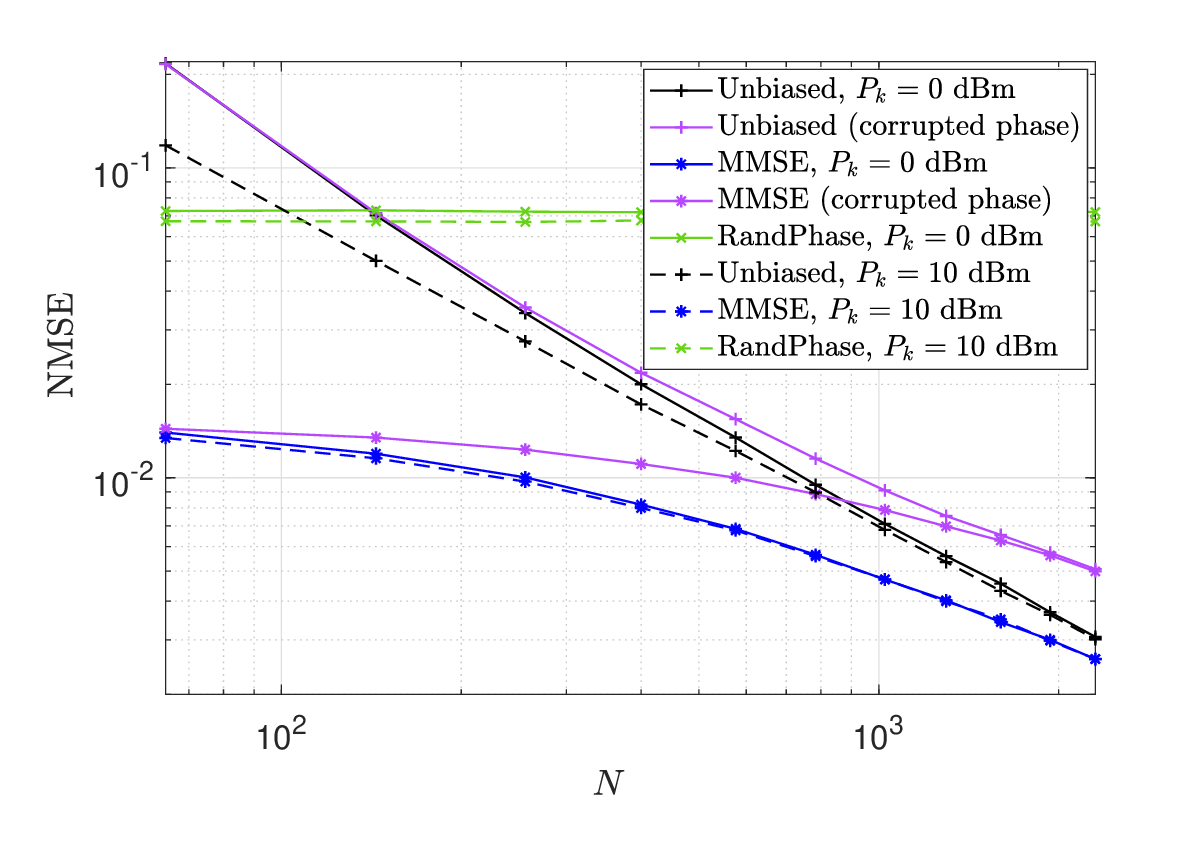}
    \caption{NMSE versus $N$ with $K=100$ and $M=5$.}
    \label{fig1} \end{figure}

\begin{appendices}
\section{Proof of Theorem \ref{theorem1}}
By substituting the RIS phase shifts in Theorem \ref{theorem1}, the mean of $\ell_{m,k}$ for $k \in \mathcal{K}_m$ is calculated as
\begin{align}\label{A1}
\mathbb{E}\left[\ell_{m,k}\right]=\frac{\sqrt{p_k}}{\lambda_m}\sum_{i=1}^{M}{\beta_{i,k}\mathbb{E}\left[\Re\left\{\mathbf{h}_{p,i,m}^H\mathbf{\Theta}_i\mathbf{h}_{i,k}\right\}\right]}.
\end{align}

1) For $i=m$, we have
\begin{align}\label{A2}
   \mathbb{E}\left[\Re\left\{\mathbf{h}_{p,m,m}^H\mathbf{\Theta}_m\mathbf{h}_{m,k}\right\}\right]&=
   \Re\left\{\mathbb{E}\left[\mathbf{h}_{p,m,m}^H\mathbf{\Theta}_m\mathbf{h}_{m,k}\right]\right\}\nonumber \\
   &=\Re\left\{\mathbb{E}\left[\sum_{n=1}^{N}{\left|h_{p,m,m,n}^\ast\right|h_{m,k,n}\frac{\sum_{k\in\mathcal{K}_m} h_{m,k,n}^\ast}{\left|\sum_{k\in\mathcal{K}_m} h_{m,k,n}^\ast\right|}}\right]\right\}\nonumber \\
   &\overset{\text{(a)}}{=}\frac{\sqrt\pi N}{2}\Re\left\{\mathbb{E}\left[h_{m,k,n}\frac{\sum_{k\in\mathcal{K}_m} h_{m,k,n}^\ast}{\left|\sum_{k\in\mathcal{K}_m} h_{m,k,n}^\ast\right|}\right]\right\}\nonumber \\
   &\overset{\text{(b)}}{=}\frac{\sqrt\pi N}{2\left|\mathcal{K}_m\right|}\Re\left\{\mathbb{E}\left[\sum_{k\in\mathcal{K}_m} h_{m,k,n}\frac{\sum_{k\in\mathcal{K}_m} h_{m,k,n}^\ast}{\left|\sum_{k\in\mathcal{K}_m} h_{m,k,n}^\ast\right|}\right]\right\}\nonumber \\
   &=\frac{\sqrt\pi N}{2\left|\mathcal{K}_m\right|}\Re\left\{\mathbb{E}\left[\left|\sum_{k\in\mathcal{K}_m} h_{m,k,n}\right|\right]\right\}\nonumber \\
   &\overset{\text{(c)}}{=}\frac{\sqrt\pi N}{2\left|\mathcal{K}_m\right|}\frac{\sqrt{\left|\mathcal{K}_m\right|\pi}}{2}=\frac{\pi N}{4\sqrt{\left|\mathcal{K}_m\right|}},
\end{align}
where $\mathrm{(a)}$ is due to the independence between different channels and $\mathbb{E}\left[|h_{p,m,m,n}^*|\right]=\frac{\sqrt{\pi}}{2}$~\cite{xjd}, $\mathrm{(b)}$ follows from the identically distributed characteristic of $h_{m,k,n}\frac{\sum_{k\in\mathcal{K}_m} h_{m,k,n}^\ast}{\left|\sum_{k\in\mathcal{K}_m} h_{m,k,n}^\ast\right|}$ for different $k\in\mathcal{K}_m$, and $({\rm c})$ comes from the fact that $\sum_{k\in\mathcal{K}_m} h_{m,k,n}\sim\mathcal{CN}(0,\left|\mathcal{K}_m\right|)$~\cite{sw1,sw2}. 

2) For $i\neq m$, we have
\begin{align}\label{A3}
    \mathbb{E}\left[\Re\left\{\mathbf{h}_{p,i,m}^H\mathbf{\Theta}_i\mathbf{h}_{i,k}\right\}\right]&=\Re\left\{\mathbb{E}\left[\mathbf{h}_{p,i,m}^H\mathbf{\Theta}_i\mathbf{h}_{i,k}\right]\right\}\nonumber \\
    & =\Re \left\{\mathbb{E}\left[\sum_{n=1}^{N}{h_{p,i,m,n}^\ast {\rm e}^{-j\angle h_{p,i,i,n}^\ast}h_{i,k,n}\frac{\sum_{k^\prime\in\mathcal{K}_i} h_{i,k^\prime,n}^\ast}{\left|\sum_{k^\prime\in\mathcal{K}_i} h_{i,k^\prime,n}^\ast\right|}}\right]\right\}\nonumber \\
    &\overset{\text{(d)}}{=}N\cdot \Re \left\{\mathbb{E}\left[h_{p,i,m,n}^\ast\right]\mathbb{E}\left[{\rm e}^{-j\angle h_{p,i,i,n}^\ast}\right]\mathbb{E}\left[h_{i,k,n}\right]\mathbb{E}\left[\frac{\sum_{k^\prime\in\mathcal{K}_i} h_{i,k^\prime,n}^\ast}{\left|\sum_{k^\prime\in\mathcal{K}_i} h_{i,k^\prime,n}^\ast\right|}\right]\right\}\nonumber \\
    &=0,
\end{align}
where $\mathrm{(d)}$ exploits the independence of $\mathbf{h}_{p,i,m}$, $\mathbf{h}_{p,i,i}$, $\mathbf{h}_{i,k}$, and $\mathbf{h}_{i,k^\prime}$, and the last equality comes from $h_{p,i,m,n}\sim\mathcal{CN}(0,1)$.

Then, by substituting (\ref{A2}) and (\ref{A3}) into (\ref{A1}), it yields
\begin{align}\label{A4}
\mathbb{E}\left[\ell_{m,k}\right]=\frac{\sqrt{p_k}\beta_{m,k}}{\lambda_m}\frac{\pi N}{4\sqrt{\left|\mathcal{K}_m\right|}}>0.
\end{align}

In addition, the mean of $\ell_{m,k^\prime}$ for $k^\prime\in\mathcal{K}_{m^\prime}$ and $m^\prime\neq m$ is calculated as
\begin{align}\label{A5}
\mathbb{E}\left[\ell_{m,k^\prime}\right]=\frac{\sqrt{p_{k^\prime}}}{\lambda_m}\sum_{i=1}^{M}{\beta_{i,k^\prime}\mathbb{E}\left[\Re\left\{\mathbf{h}_{p,i,m}^H\mathbf{\Theta}_i\mathbf{h}_{i,k^\prime}\right\}\right]}.
\end{align}

1) For $i=m$, we have
\begin{align}\label{A6}
    \mathbb{E}\left[\Re\left\{\mathbf{h}_{p,m,m}^H\mathbf{\Theta}_m\mathbf{h}_{m,k^\prime}\right\}\right]&=\Re\left\{\mathbb{E}\left[\mathbf{h}_{p,m,m}^H\mathbf{\Theta}_m\mathbf{h}_{m,k^\prime}\right]\right\}\nonumber \\
    &=\Re\left\{\mathbb{E}\left[\sum_{n=1}^{N}{\left|h_{p,m,m,n}^\ast\right|h_{m,k^\prime,n}\frac{\sum_{k\in\mathcal{K}_m} h_{m,k,n}^\ast}{\left|\sum_{k\in\mathcal{K}_m} h_{m,k,n}^\ast\right|}}\right]\right\}\nonumber \\
    &=N\cdot \Re\left\{\mathbb{E}\left[\left|h_{p,m,m,n}^\ast\right|\right]\mathbb{E}\left[h_{m,k^\prime,n}\right]\mathbb{E}\left[\frac{\sum_{k\in\mathcal{K}_m} h_{m,k,n}^\ast}{\left|\sum_{k\in\mathcal{K}_m} h_{m,k,n}^\ast\right|}\right]\right\}\nonumber \\
    &=0.\end{align}

2) For $i=m^\prime$, we have
\begin{align}\label{A7}
    \mathbb{E}\!\left[\Re\!\left\{\mathbf{h}_{p,m^\prime,m}^H\mathbf{\Theta}_{m^\prime}\mathbf{h}_{m^\prime,k^\prime}\right\}\right]&\!=\!\Re\left\{\mathbb{E}\left[\mathbf{h}_{p,m^\prime,m}^H\mathbf{\Theta}_{m^\prime}\mathbf{h}_{m^\prime,k^\prime}\right]\right\}\nonumber \\
    &\!=\!\Re\left\{\mathbb{E}\left[\sum_{n=1}^{N}{h_{p,m^\prime,m,n}^\ast {\rm e}^{-j\angle h_{p,m^\prime,m^\prime,n}^\ast}h_{m^\prime,k^\prime,n}\frac{\sum_{k^\prime\in\mathcal{K}_{m^\prime}} h_{m^\prime,k^\prime,n}^\ast}{\left|\sum_{k^\prime\in\mathcal{K}_{m^\prime}} h_{m^\prime,k^\prime,n}^\ast\right|}}\right]\right\}\nonumber\\
    &\!=\!N\cdot \Re\!\left\{\!\mathbb{E}\!\left[h_{p,m^\prime,m,n}^\ast\right]\!\mathbb{E}\!\left[{\rm e}^{-j\angle h_{p,m^\prime,m^\prime,n}^\ast}\right]\!\mathbb{E}\!\!\left[h_{m^\prime,k^\prime,n}\frac{\sum_{k^\prime\in\mathcal{K}_{m^\prime}} h_{m^\prime,k^\prime,n}^\ast}{\left|\sum_{k^\prime\in\mathcal{K}_{m^\prime}} h_{m^\prime,k^\prime,n}^\ast\right|}\!\right]\!\right\}\nonumber \\
    &\!=\!0.
\end{align}

3) For $i\neq m$ and $i\neq m^\prime$, we have
\begin{align}\label{A8}
    \mathbb{E}\left[\Re\left\{\mathbf{h}_{p,i,m}^H\mathbf{\Theta}_i\mathbf{h}_{i,k^\prime}\right\}\right]&=\Re\left\{\mathbb{E}\left[\mathbf{h}_{p,i,m}^H\mathbf{\Theta}_i\mathbf{h}_{i,k^\prime}\right]\right\}\nonumber \\
    &=\Re\left\{\mathbb{E}\left[\sum_{n=1}^{N}{h_{p,i,m,n}^\ast {\rm e}^{-j\angle h_{p,i,i,n}^\ast}h_{i,k^\prime,n}\frac{\sum_{k^{\prime\prime}\in\mathcal{K}_i} h_{i,k^{\prime\prime},n}^\ast}{\left|\sum_{k^{\prime \prime}\in\mathcal{K}_i} h_{i,k^{\prime \prime},n}^\ast\right|}}\right]\right\}\nonumber \\
    &=N\cdot \Re\left\{\mathbb{E}\left[h_{p,i,m,n}^\ast\right]\mathbb{E}\left[{\rm e}^{-j\angle h_{p,i,i,n}^\ast}\right]\mathbb{E}\left[h_{i,k^\prime,n}\right]\mathbb{E}\left[\frac{\sum_{k^{\prime \prime}\in\mathcal{K}_i} h_{i,k^{\prime \prime},n}^\ast}{\left|\sum_{k^{\prime \prime}\in\mathcal{K}_i} h_{i,k^{\prime \prime},n}^\ast\right|}\right]\right\}\nonumber \\
    &=0.
\end{align}

Then, by substituting (\ref{A6}), (\ref{A7}) and (\ref{A8}) into (\ref{A5}), it yields 
\begin{align}\label{A9}
\mathbb{E}\left[\ell_{m,k^\prime}\right]=0.
\end{align}

\section{Proof of Proposition~\ref{proposition1}}
By directly applying Theorem \ref{theorem1} and substituting the parameters in Proposition~\ref{proposition1} into the expectation of the estimated global gradient $\hat{\mathbf{g}}_{m,t}$, we have 
\begin{align}\label{B1}
\mathbb{E}\left[{\hat{\mathbf{g}}}_{m,t}\right]=&\sum_{k\in\mathcal{K}_m}{\mathbb{E}\left[\ell_{m,k}\right]\bar{\mathbf{g}}_{m,t,k}}+ \sum_{k\in\mathcal{K}_m }\!\frac{u_{m,t,k}}{|\mathcal{K}_m|}\mathbf{1}+\sum_{\substack{1\le m^\prime\le M\\m^\prime\neq m}}\sum_{k^\prime\in\mathcal{K}_{m^\prime}}{\mathbb{E}\left[\ell_{m,k^\prime}\right]\bar{\mathbf{g}}_{m^\prime,t,k^\prime}}\!+\mathbb{E}\left[{\bar{\mathbf{z}}}_{m,t}\right]\nonumber \\
% =&\sum_{k\in\mathcal{K}_m}{\frac{\sqrt{p_k}\beta_{m,k}}{\lambda_m}\frac{\pi N}{4\sqrt{\left|\mathcal{K}_m\right|}}{\bar{\mathbf{g}}}_{m,t,k}}+\sum_{k\in\mathcal{K}_m}\frac{u_{m,t,k}}{\left|\mathcal{K}_m\right|}\mathbf{1}\nonumber \\
=&\sum_{k\in\mathcal{K}_m}{\frac{\sigma_{m,t,k}}{\left|\mathcal{K}_m\right|}\frac{1}{\sigma_{m,t,k}}\left(\mathbf{g}_{m,t,k}-u_{m,t,k}\mathbf{1}\right)}+\sum_{k\in\mathcal{K}_m}\frac{u_{m,t,k}}{\left|\mathcal{K}_m\right|}\mathbf{1}\nonumber \\
=&\frac{1}{\left|\mathcal{K}_m\right|}\sum_{k\in\mathcal{K}_m}\mathbf{g}_{m,t,k}=\mathbf{g}_{m,t}.
\end{align}
Therefore, the expectation of the estimated global gradient $\hat{\mathbf{g}}_{m,t}$ is equal to the ground-truth global gradient $\mathbf{g}_{m,t}$ for $m\in\left[M\right]$, which ensures the unbiasedness of gradient transmission~\cite{unbia}. This completes the proof.

\section{Proof of Proposition~\ref{proposition2}}

To begin with, we formulate the MSE of gradient estimation for cluster $m$ as
\begin{align}\label{mse}
    \text{MSE}_m &\!=\!\mathbb{E}\!\left[\Vert\mathbf{g}_{m,t}-\hat{\mathbf{g}}_{m,t} \Vert^2\right]\nonumber \\
    &\!=\!\mathbb{E}\!\left[\left\Vert\sum_{k\in\mathcal{K}_m}{\!{\ell_{m,k}}\bar{\mathbf{g}}_{m,t,k}}\!+\!\sum_{k\in\mathcal{K}_m }\!\frac{u_{m,t,k}}{|\mathcal{K}_m|}\mathbf{1}\!+\!\sum_{\substack{1\le m^\prime\le M\\m^\prime\neq m}}\sum_{k'\in\mathcal{K}_m\prime}\!{\ell_{m,k'}\bar{\mathbf{g}}_{m^\prime,t,k'}}\!+\!\bar{\mathbf{z}}_{m,t}\!-\!\sum_{k\in\mathcal{K}_m} \!\frac{1}{|\mathcal{K}_m|}\mathbf{g}_{m,t,k} \right \Vert^2\right]\nonumber \\
    &\!\overset{\text{(a)}}{=}\!\mathbb{E}\!\left[\left\Vert\sum_{k\in\mathcal{K}_m}{\left(\ell_{m,k}-\frac{\sigma_{m,t,k}}{|\mathcal{K}_m|}\right)\bar{\mathbf{g}}_{m,t,k}}+ \sum_{\substack{1\le m^\prime\le M\\m^\prime\neq m}}\sum_{k'\in\mathcal{K}_m\prime}{\ell_{m,k'}\bar{\mathbf{g}}_{m^\prime,t,k'}}+\bar{\mathbf{z}}_{m,t} \right \Vert^2\right]\nonumber \\
    &\!\overset{\text{(b)}}{=}\!\sum_{k\in\mathcal{K}_m}\left(\ell_{m,k}-\frac{\sigma_{m,t,k}}{|\mathcal{K}_m|}\right)^2 \sigma_{m,t,k}^2 D + \sum_{\substack{1\le m^\prime\le M\\m^\prime\neq m}}\sum_{k'\in\mathcal{K}_m\prime} \ell_{m,k^\prime}\sigma_{m^\prime,t,k^\prime}^2D +\frac{\sigma^2 D}{\lambda_m^2}\nonumber \\
    &\!\overset{\text{(c)}}{=}\!\left(\sum_{i=1}^M \sum_{k\in \mathcal{K}_i}  p_k \bar{h}_{m,k}^2\sigma_{i,t,k}^2 D\!+\!\sigma^2 D\right)\frac{1}{\lambda_m^2}\!-\!2\left(\sum_{k\in\mathcal{K}_m} \frac{\sqrt{p_k} \bar{h}_{m,k}\sigma_{m,t,k}^3D}{|\mathcal{K}_m|}\right)\frac{1}{\lambda_m}\!+\!\sum_{k\in\mathcal{K}_m} \frac{\sigma_{m,t,k}^4D}{|\mathcal{K}_m|^2},
\end{align}
where (a) is due to the definition of $\bar{g}_{m,t,k}$, (b) exploits the statistics of $\bar{\mathbf{g}}_{m,t,k}$ and $\bar{\mathbf{z}}_{m,t}$, and (c) comes from the definition of $\ell_{m,k}$. Note that the optimization of denoising factor is an unconstrained problem. For any given power control, we derive the optimal denoising factor, $\lambda_m$, by checking the following equality
\begin{align}
    &\frac{\partial\text{MSE}_m}{\partial\lambda_m }=-2\left(\sum_{i=1}^M \sum_{k\in \mathcal{K}_i}  p_k \bar{h}_{m,k}^2\sigma_{i,t,k}^2 D + \sigma^2 D\right)\frac{1}{\lambda_m^3}+2\left(\sum_{k\in\mathcal{K}_m} \frac{\sqrt{p_k} \bar{h}_{m,k}\sigma_{m,t,k}^3D}{|\mathcal{K}_m|}\right)\frac{1}{\lambda_m^2}=0\nonumber\\
    &\Rightarrow \lambda_m^* = |\mathcal{K}_m| \frac{\sum_{i=1}^M \sum_{k\in \mathcal{K}_i}  p_k \bar{h}_{m,k}^2\sigma_{i,t,k}^2  + \sigma^2 }{\sum_{k\in\mathcal{K}_m} \sqrt{p_k} \bar{h}_{m,k}\sigma_{m,t,k}^3},
\end{align}
and the proof completes.

As for the optimization of power control, substituting the optimal $\lambda_m^*$ into (\ref{mse}), we rewrite $\text{MSE}_m$ as 
\begin{align}
    \text{MSE}_m&=\sum_{k\in\mathcal{K}_m} \frac{\sigma_{m,t,k}^4D}{|\mathcal{K}_m|^2}-\frac{\left(\sum_{k\in\mathcal{K}_m} \sqrt{p_k} \bar{h}_{m,k}\sigma_{m,t,k}^3\right)^2 D}{|\mathcal{K}_m|^2\left(\sum_{i=1}^M \sum_{k\in \mathcal{K}_i}  p_k \bar{h}_{m,k}^2\sigma_{i,t,k}^2  + \sigma^2\right)}\nonumber \\
    &=\sum_{k\in\mathcal{K}_m} \frac{\sigma_{m,t,k}^4D}{|\mathcal{K}_m|^2}-\frac{\left(\mathbf{p}^T\mathbf{b}_m\right)^2D}{\mathbf{p}^T \mathbf{A}_m \mathbf{p}+|\mathcal{K}_m|^2\sigma^2}\nonumber \\
    &=\sum_{k\in\mathcal{K}_m} \frac{\sigma_{m,t,k}^4D}{|\mathcal{K}_m|^2}-D\frac{\mathbf{p}^T\mathbf{B}_m \mathbf{p}}{\mathbf{p}^T \mathbf{A}_m \mathbf{p}+|\mathcal{K}_m|^2\sigma^2},
\end{align}
where $\mathbf{p}\triangleq [\sqrt{p_1},\sqrt{p_2},\cdots,\sqrt{p_K}]^T$, $\mathbf{b}_m \triangleq \sum_{k\in\mathcal{K}_m} \bar{h}_{m,k}\sigma_{m,t,k}^3\mathbf{e}_k$, $\mathbf{A}_m \triangleq |\mathcal{K}_m|^2\text{diag}\left\{\bar{h}_{m,k}^2 \sigma_{i,t,k}^2 \right\}$, $\mathbf{B}_m\triangleq \mathbf{b}_m \mathbf{b}_m^T$, and $\mathbf{e}_k$ is  the Kronecker delta vector with $[\mathbf{e}_k]_k=1$. Now, we formulate an equivalent power control optimization problem for minimizing the sum MSE as
\begin{align}\label{problem}
    \mathop{\text{maximize}}_{\mathbf{p}}&\quad \sum_{m=1}^M \frac{\mathbf{p}^T\mathbf{B}_m \mathbf{p}}{\mathbf{p}^T \mathbf{A}_m \mathbf{p}+|\mathcal{K}_m|^2\sigma^2}\nonumber \\
    \text{subject to}&\quad [\mathbf{p}]_k\leq \sqrt{P_k},\enspace \forall k.
\end{align}
It worth noting that the problem in (\ref{problem}) is known as the sum of quadratic ratios maximization, which has been addressed in existing works via branch and bound \cite{opt1}, harmony search method \cite{opt2} and semidefinite relaxation (SDR) technique \cite{opt3,hzy1,hzy2}.

\end{appendices}
% \bibliographystyle{IEEEtran}
% \bibliography{SCMApaper}

\end{document}